\documentclass[reqno]{amsart}
\usepackage{bm}
\usepackage{amsmath}
\usepackage{amsfonts,amssymb}
\usepackage[dvips]{graphicx}
\usepackage{epsfig}
\usepackage{color}
\usepackage{ulem}
\newtheorem{theorem}{Theorem}[section]
\newtheorem{remark}{Remark}[section]
\newtheorem{lemma}[theorem]{Lemma}
\newtheorem{definition}{Definition}[section]

\newtheorem{proposition}[theorem]{Proposition}
\def\labelenumi{\theenumi}

\numberwithin{equation}{section}
\begin{document}

\title[Hotelling model]
{Necessary and sufficient condition for equilibrium of the Hotelling model on a circle}
\author{Satoshi Hayashi${}^{\dagger}$}
\address{${}^{\dagger\ddagger}$Department of Mathematics Education, 
	Faculty of Education, Gifu University, 1-1 Yanagido, Gifu
	Gifu 501-1193 Japan.}
\email{${}^{\dagger}$x1131023@edu.gifu-u.ac.jp}

\author{Naoki Tsuge${}^{\ddagger}$}

\email{${}^{\ddagger}$tuge@gifu-u.ac.jp}
\thanks{
	N. Tsuge's research is partially supported by Grant-in-Aid for Scientific 
	Research (C) 17K05315, Japan.
}


\date{}

\maketitle

\begin{abstract}
We are concerned with a model of vendors competing to sell a homogeneous product to customers spread evenly along a circular city. 
This model is based on Hotelling's celebrated paper in 1929.  Our aim in this paper is to show a necessary and sufficient condition for the equilibrium, 
which describes geometric properties of the equilibrium. To achieve this, we first formulate the model mathematically. 
Next, we prove that the condition holds if and only if vendors are equilibrium. 
\end{abstract}


\section{Introduction}
We study a model in which a circular city lies on a circle with circumference length $1$ and 
customers are uniformly distributed with density $1$ along this circle.  
We consider $n$ vendors moving on this circle. 
The model governs a competition between vendors on a simple closed curve.
Let the location of 
the vendor $k\quad(k=1,2,3,\ldots,n)$ be $x_k\in[0,1]$.  We assume that $x_1\leq x_2\leq \cdots\leq x_n$ and  denote the location of $n$ vendors $(x_1,x_2,x_3,\ldots,x_n)$ by ${\bm x}$. Since we study the 
competition between vendors, we consider $n\geq2$ in particular.
The price of one unit of product for each vendor is 
identical. Moreover, we assume the following. 
\vspace*{1ex}

If there exist $l\quad(l=1,2,3,\ldots,n)$ vendors nearest to 
a customer, the customer  
 purchases $1/l$ unit of product per unit of time from each of the $l$ vendors respectively. \vspace*{-1ex}

Every vendor then seeks a location to maximize his profit.

We then represent the profit of  vendor $k$ per unit of time by a mathematical notation.
We first denote a distance on the circle between $x,y\in[0,1]$ by 
$d(x,y)=\min\{|x-y+1|,|x-y|,|x-y-1|\}$. 
\begin{remark}We regard $x,y\in[0,1]$ as points on a circle with circumference length $1$. Then 
	$d(x,y)$ represents the length of bold-faced arcs in the following figure. 
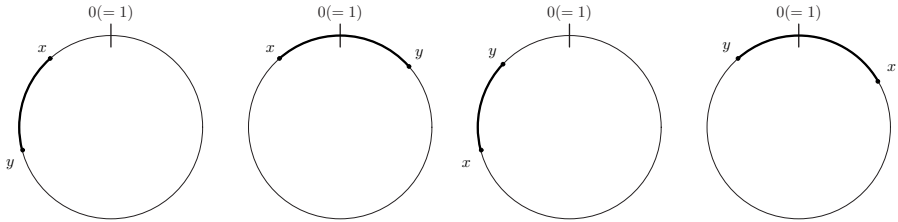
\begin{center}
	\scalebox{0.6}{
{\unitlength 0.1in%
\begin{picture}(77.7000,18.6500)(0.3000,-20.0000)%
%
\special{pn 8}%
\special{ar 1000 1200 800 800 0.0000000 6.2831853}%
%
\special{pn 8}%
\special{ar 3000 1200 800 800 0.0000000 6.2831853}%
%
\special{pn 8}%
\special{ar 5000 1200 800 800 0.0000000 6.2831853}%
%
\special{pn 8}%
\special{ar 7000 1200 800 800 0.0000000 6.2831853}%
%
\special{pn 13}%
\special{pa 1000 300}%
\special{pa 1000 500}%
\special{fp}%
\special{pa 3000 500}%
\special{pa 3000 300}%
\special{fp}%
\special{pa 5000 300}%
\special{pa 5000 500}%
\special{fp}%
\special{pa 7000 500}%
\special{pa 7000 300}%
\special{fp}%
\put(10.0000,-2.0000){\makebox(0,0){$0(=1)$}}%
\put(30.0000,-2.0000){\makebox(0,0){$0(=1)$}}%
\put(50.0000,-2.0000){\makebox(0,0){$0(=1)$}}%
\put(70.0000,-2.0000){\makebox(0,0){$0(=1)$}}%
%
\special{pn 4}%
\special{sh 1}%
\special{ar 470 600 16 16 0 6.2831853}%
%
\special{pn 4}%
\special{sh 1}%
\special{ar 230 1400 16 16 0 6.2831853}%
%
\special{pn 4}%
\special{sh 1}%
\special{ar 2470 600 16 16 0 6.2831853}%
%
\special{pn 4}%
\special{sh 1}%
\special{ar 3600 670 16 16 0 6.2831853}%
%
\special{pn 4}%
\special{sh 1}%
\special{ar 4420 650 16 16 0 6.2831853}%
%
\special{pn 4}%
\special{sh 1}%
\special{ar 4230 1400 16 16 0 6.2831853}%
%
\special{pn 4}%
\special{sh 1}%
\special{ar 6470 600 16 16 0 6.2831853}%
%
\special{pn 4}%
\special{sh 1}%
\special{ar 7690 800 16 16 0 6.2831853}%
\put(4.0000,-5.1000){\makebox(0,0){$x$}}%
\put(1.2000,-15.2000){\makebox(0,0){$y$}}%
\put(23.9000,-5.1000){\makebox(0,0){$x$}}%
\put(36.9000,-5.6000){\makebox(0,0){$y$}}%
\put(43.3000,-5.5000){\makebox(0,0){$y$}}%
\put(41.0000,-15.2000){\makebox(0,0){$x$}}%
\put(63.7000,-5.0000){\makebox(0,0){$y$}}%
\put(78.1000,-6.8000){\makebox(0,0){$x$}}%
%
\special{pn 20}%
\special{ar 1000 1200 800 800 2.8913921 3.9907341}%
%
\special{pn 20}%
\special{ar 3000 1200 800 800 3.9821495 5.5551953}%
%
\special{pn 20}%
\special{ar 5000 1200 800 800 2.8868344 3.8974136}%
%
\special{pn 20}%
\special{ar 7000 1200 800 800 3.9857466 5.7564348}%
\end{picture}}%
}
	\label{define}
\end{center}\vspace{5mm}	
That is, $d(x,y)$ is the shortest distance between $x$ and $y$ on the circle.
\end{remark}
Given a vector ${\bm \xi}=(\xi_1,\xi_2,\ldots,\xi_n)\in[0,1]^n$ and $0\leq y\leq1$, 
we then define a set $\displaystyle S({\bm \xi},y)=\{j\in\{1,2,3,\ldots,n\}:d(\xi_j,y)=\min_{i}d(\xi_i,y)\}$.
By using a density function 
\begin{align*}
\rho_k({\bm \xi},y)=
\begin{cases}
\displaystyle 0&\displaystyle \text{if}\quad   d(\xi_k,y)>\min_{i}d(\xi_i,y)\\
\displaystyle \frac{1}{|S({\bm \xi},y)|}&\displaystyle \text{if}\quad d(\xi_k,y)=\min_{i}d(\xi_i,y)
\end{cases},
\end{align*}
we define 
\begin{align*}
f_k({\bm \xi})=\int^1_0\rho_k({\bm \xi},y)dy,
\end{align*}
where $|S({\bm \xi},y)|$ represents a number of elements in a set $S({\bm \xi},y)$. We call $f_k({\bm x})$ the profit of vendor $k$ per unit of time for a location ${\bm x}$. We then define equilibrium as follows.
\begin{definition}
A location ${\bm x}^*=(x^*_1,x^*_2,x^*_3,\ldots,x^*_n)\in[0,1]^n$ is called equilibrium, if 
\begin{align*}
f_k({\bm x}^*)\geq f_k(x^*_1,x^*_2,x^*_3,\ldots,x^*_{k-1},x_k,x^*_{k+1},\ldots,x^*_n)
\end{align*}holds for any $k\in\{1,2,3,\ldots,n\}$ and $x_k\in[0,1]$.
\end{definition}

We review the known results. The present model is based on the Hotelling's pioneer work \cite{H}. Although we 
consider the circle as a product space, Hotelling dealed with a finite line.  In \cite{HT},  authors show
a necessary and sufficient condition for the equilibrium of $n$ vendors on the line.

On the 
other hand,  our model was introduced by Eaton and Lipsey \cite{EL2}, who discussed with the existence of equilibrium for the model
without the price. Subsequently, taking the price  into account, Salop \cite{S} studied the model for two vendors. 
In this paper, we are concerned with $n$ vendors on the circle and investigate their equilibrium.
Our goal  is to present a necessary and sufficient condition for the equilibrium.

For convenience, we set $x_1=0,\;x_0=x_n,\;x_{-1}=x_{n-1},\;x_{-2}=x_{n-2},\;x_{n+1}=x_1,\;x_{n+2}=x_2$ and denote a interval $[x_k,x_{k+1}]$ by $I_k\;\;(k=-2,-1,0,1,\ldots,n+1)$. Then our main theorem is as follows.
\begin{theorem}\label{thm:main}${}$

	${\bm x}$ is equilibrium, if and only if the following condition \eqref{eqn:main1} hold.
	\begin{align}
	\text{We define a set }|\textcolor{black}{\bar{I}}|=\max_{l\in{\{1,2,\cdots,\textcolor{black}{n}\}}}|I_l|\nonumber.\\
	|I_k|+|I_{k+1}|\geq|\bar{I}|\quad(1\leq k\leq n)\label{eqn:main1},
	\end{align}
where $|I|$ represents the length of an interval $I$.
\end{theorem}
\begin{remark}\normalfont
From \eqref{eqn:main1}, we notice that any location becomes  equilibrium for $n=2$.
On the other hand, we find that equilibrium  for $n=3$ is $\displaystyle {\bm x^\ast}=\left(0,\frac12,\frac12\right),\left(0,0,\frac12\right)$.
\end{remark}

\section{Preliminary}
In this section, we prepare some lemmas and a proposition to prove our main theorem in a next section. 
We first consider the profit of $i$ vendors which locate at one point. We have the following lemma. 
\begin{lemma}
	We consider a location ${\bm x}=(x_1,x_2,x_3,\ldots,x_n)$ with $0=x_1\leq x_2\leq x_3\leq\dots\leq x_n$.
	We assume that $x_l<x_{l+1}=\cdots = x_k = \cdots = x_{l+i}<x_{l+i+1}\quad(n\geq2,\;l\geq0,\;l+i+1\leq n+1,\;1\leq i\leq n)$.

    \begin{align*}
    \displaystyle f_k({\bm x})=\frac1{2i}(|I_l|+|I_{l+i}|).
    \end{align*}

\end{lemma}\begin{proof}${}$

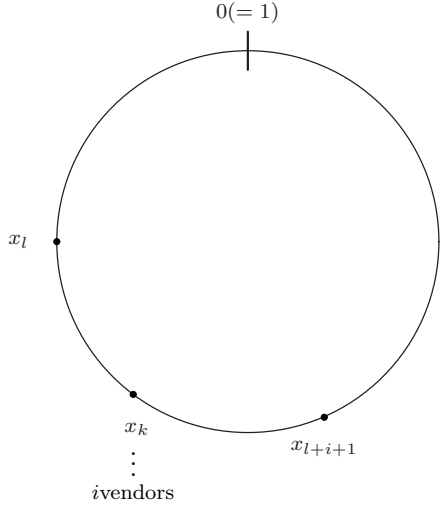
\begin{center}
	\scalebox{1.0}{
{\unitlength 0.1in%
\begin{picture}(23.8000,25.1000)(0.2000,-26.4500)%
%
\special{pn 8}%
\special{ar 1400 1400 1000 1000 0.0000000 6.2831853}%
%
\special{pn 13}%
\special{pa 1400 300}%
\special{pa 1400 500}%
\special{fp}%
\put(14.0000,-2.0000){\makebox(0,0){{\footnotesize $0(=1)$}}}%
\put(2.0000,-14.0000){\makebox(0,0){{\footnotesize $x_{l}$}}}%
\put(8.2000,-23.8000){\makebox(0,0){{\footnotesize $x_{k}$}}}%
\put(18.0000,-24.8000){\makebox(0,0){{\footnotesize $x_{l+i+1}$}}}%
%
\special{pn 4}%
\special{sh 1}%
\special{ar 1800 2320 16 16 0 6.2831853}%
%
\special{pn 4}%
\special{sh 1}%
\special{ar 400 1400 16 16 0 6.2831853}%
%
\special{pn 4}%
\special{sh 1}%
\special{ar 800 2200 16 16 0 6.2831853}%
\put(8.0000,-25.2000){\makebox(0,0){$\vdots$}}%
\put(8.0000,-27.1000){\makebox(0,0){{\footnotesize $i$vendors}}}%
\end{picture}}%
}
    \label{lem2.1}
\end{center}\vspace{5mm}


We have $f_k({\bm x})=\dfrac{1}{i}\left(\dfrac{1}{2}|I_l|+\dfrac{1}{2}|I_{l+i}|\right)=\dfrac{1}{2i}\left(|I_l|+|I_{l+i}|\right)$.

\end{proof}

Next, the following proposition play an important role.

\begin{proposition}\label{pro:1}
	If the location  of $n$ vendors ${\bm x}=(x_1,x_2,x_3,\ldots,x_n)$ with $0=x_1\leq x_2\leq x_3\leq\dots\leq x_n\quad(n\geq2)$ is equilibrium,  
	the following holds.
	\begin{align}
	\text{No more than 2 vendors can occupy a location.}
	\label{eqn:pro1}
	\end{align}
\end{proposition}

\begin{proof}
${}$	
	
	We prove that $ {\bm x} $ is not equilibrium, provided that $i\;(3\leq i\leq n)$ vendors occupy at a point. We assume that $x_l<x_{l+1}=\cdots = x_k = \cdots = x_{l+i}<x_{l+i+1}\quad(l\geq0,\;l+i+1\leq n+1,\;3\leq i\leq n)$.
	Therefore there exist $x_l$ and $x_{l+i+1}$ at least one respectively. We notice 
	that $x_l<x_k<x_{l+i+1}$ and there exists no vendor on $(x_l,x_k)$ and $(x_k,x_{l+i+1})$.\\
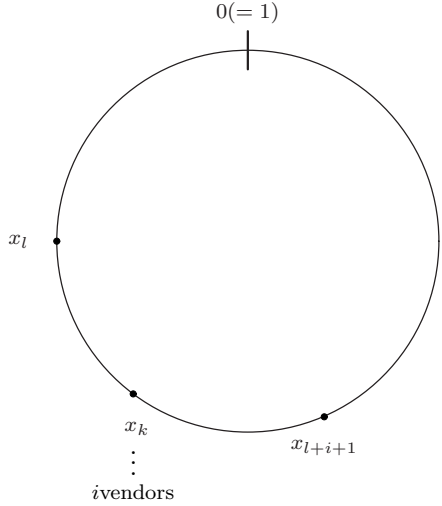
    \begin{center}
    	\scalebox{1.0}{
{\unitlength 0.1in%
\begin{picture}(23.8000,25.1000)(0.2000,-26.4500)%
%
\special{pn 8}%
\special{ar 1400 1400 1000 1000 0.0000000 6.2831853}%
%
\special{pn 13}%
\special{pa 1400 300}%
\special{pa 1400 500}%
\special{fp}%
\put(14.0000,-2.0000){\makebox(0,0){{\footnotesize $0(=1)$}}}%
\put(2.0000,-14.0000){\makebox(0,0){{\footnotesize $x_{l}$}}}%
\put(8.2000,-23.8000){\makebox(0,0){{\footnotesize $x_{k}$}}}%
\put(18.0000,-24.8000){\makebox(0,0){{\footnotesize $x_{l+i+1}$}}}%
%
\special{pn 4}%
\special{sh 1}%
\special{ar 1800 2320 16 16 0 6.2831853}%
%
\special{pn 4}%
\special{sh 1}%
\special{ar 400 1400 16 16 0 6.2831853}%
%
\special{pn 4}%
\special{sh 1}%
\special{ar 800 2200 16 16 0 6.2831853}%
\put(8.0000,-25.2000){\makebox(0,0){$\vdots$}}%
\put(8.0000,-27.1000){\makebox(0,0){{\footnotesize $i$vendors}}}%
\end{picture}}%
}
    	\label{pro2.2}
    \end{center}\vspace{5mm}

	If $|I_l|\geq|I_{l+i}|$, we notice that $f_k({\bm x})=\dfrac{1}{2i}(|I_l|+|I_{l+i}|)\leq \dfrac{1}{6}(|I_l|+|I_{l+i}|)\leq\dfrac{1}{6}(|I_l|+|I_l|)=\dfrac{1}{3}|I_l|$.
		Setting $x_k'\in{(x_l,x_{l+1})}$, we have
		$f_k(x_1,\cdots,x_{k-1},x_k',x_{k+1},\cdots,x_n)=\dfrac{1}{2}(|[x_l,x_k']|+|[x_k',x_k]|)
		=\dfrac{1}{2}|I_l|>f_k({\bm x})$.

		For the other case $|I_l|<|I_{l+i}|$, from the symmetry, we can similarly show that 
		there exists $x_k'$ such that $f_k(x_1,\cdots,x_{k-1},x_k',x_{k+1},\cdots,x_n)>f_k({\bm x})$.

        We can complete the proof of \eqref{eqn:pro1}.
\end{proof}

Finally, we compare a location after the movement of a vendor with the original one. To do this, we introduce the following notation.

For a given location  of $n$ vendors ${\bm x}=(x_1,x_2,x_3,\ldots,x_n)$ with $(0=)\;x_1\leq x_2\leq \cdots\leq x_n$, we move vendor $k$ from $x_k$ to any fixed point in 
$A\subset [0,1]$. We denote the resultant location by $x_k\rightarrow A$. We notice that $x_k\rightarrow A$ represents the following vector
\begin{align*}
(x_1,x_2,,\ldots,x_{k-1},x'_k,x_{k+1},\ldots,x_n)
,
\end{align*}
where $x'_k$ is a location of vendor $k$ after movement and $x'_k\in A$.

\section{Proof of Theorem \ref{thm:main}} We are now position to prove our main theorem. We divide the proof into two cases, (I) $n=2$ and (II) $n\geq3$.
\section*{Proof of Theorem \ref{thm:main} {\rm (I)}}
We first treat with the case where $n=2$ and prove that every ${\bm x}$ is equilibrium.
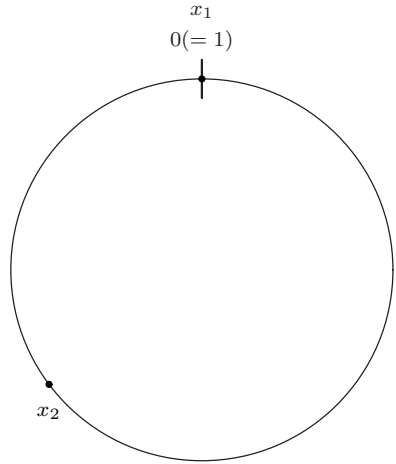
\begin{center}
	\scalebox{1.0}{
{\unitlength 0.1in%
\begin{picture}(22.9500,24.1500)(3.0500,-26.0000)%
%
\special{pn 8}%
\special{ar 1600 1600 1000 1000 0.0000000 6.2831853}%
%
\special{pn 13}%
\special{pa 1600 500}%
\special{pa 1600 700}%
\special{fp}%
\put(16.0000,-4.0000){\makebox(0,0){\footnotesize $0(=1)$}}%
%
\put(16.0000,-4.0000){\makebox(0,0)[lb]{}}%
\put(16.0000,-2.5000){\makebox(0,0){\footnotesize $x_1$}}%
%
\special{pn 4}%
\special{sh 1}%
\special{ar 1600 600 16 16 0 6.2831853}%
%
\special{pn 4}%
\special{sh 1}%
\special{ar 800 2200 16 16 0 6.2831853}%
\put(8.0000,-23.5000){\makebox(0,0){\footnotesize $x_2$}}%
\end{picture}}%
}
	\label{均衡点 of 2vendor}
\end{center}\vspace{5mm}

\begin{proof}
	
\begin{enumerate}
	\renewcommand{\labelenumi}{(\roman{enumi})}
	\item $x_1\neq x_2$\\
We notice that $f_1({\bm x})=f_2({\bm x})=\dfrac{1}{2}$ in this case.

	\item $x_1=x_2$
	

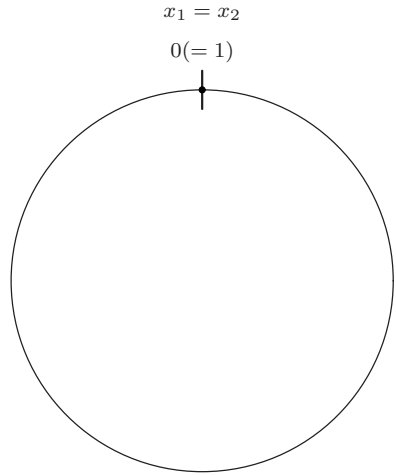
\begin{center}
	\scalebox{1.0}{
{\unitlength 0.1in%
\begin{picture}(20.0000,24.6500)(8.0000,-28.0000)%
%
\special{pn 8}%
\special{ar 1800 1800 1000 1000 0.0000000 6.2831853}%
%
\special{pn 13}%
\special{pa 1800 700}%
\special{pa 1800 900}%
\special{fp}%
\put(18.0000,-6.0000){\makebox(0,0){\footnotesize $0(=1)$}}%
\put(18.0000,-4.0000){\makebox(0,0){\footnotesize $x_1=x_2$}}%
%
\special{pn 4}%
\special{sh 1}%
\special{ar 1800 800 16 16 0 6.2831853}%
\end{picture}}%
}
	\label{均衡点 of 2vendor}
\end{center}\vspace{5mm}
We notice that $f_1({\bm x})=f_2({\bm x})=\dfrac{1}{2}$ in this case.

\end{enumerate}
It follows that $f_1({\bm x})=f_2({\bm x})=\dfrac{1}{2}$ for every ${\bm x}$. 
As a result, we showed that $f_k({\bm x})\geq f_k(x_k\rightarrow [0,1])$\quad($k=1,2$).
Forthermore, it holds that $|I_1|+|I_2|\geq|I_1|$ and $|I_1|+|I_2|\geq|I_2|$. 
Therefore, we can prove \eqref{eqn:main1}.

\end{proof}

\section*{Proof of Theorem \ref{thm:main} {\rm (II)}}
Finally, we are concerned with the case where $n\geq3$.
\begin{proof}
	First,  we show that \eqref{eqn:main1} is a sufficient condition for equilibrium. \\
	We recall that we defined $|\bar{I}|=\displaystyle\max_{l\in{\{1,2,\cdots,n\}}}|I_l|$.
	
	We show that $f_k({\bm x})\geq \dfrac{1}{2}|\bar{I}|$.
		\begin{enumerate}
		\renewcommand{\labelenumi}{(\roman{enumi})}
		\item $x_{k-1}\neq x_k$ and $x_k\neq x_{k+1}$\\
		      We notice that $f_k({\bm x})=\dfrac12(|I_{k-1}|+|I_k|)\geq\dfrac12|\bar{I}|$.
		      
		      \begin{center}
		      	\scalebox{1.0}{
{\unitlength 0.1in%
\begin{picture}(27.9000,22.7000)(-1.9000,-26.0500)%
%
\put(16.0000,-16.0000){\makebox(0,0)[lb]{}}%
%
\special{pn 8}%
\special{ar 1600 1600 1000 1000 0.0000000 6.2831853}%
%
\special{pn 13}%
\special{pa 1600 500}%
\special{pa 1600 700}%
\special{fp}%
\put(16.0000,-4.0000){\makebox(0,0){\footnotesize $0(=1)$}}%
%
\special{pn 4}%
\special{sh 1}%
\special{ar 600 1600 16 16 0 6.2831853}%
\special{sh 1}%
\special{ar 1000 2400 16 16 0 6.2831853}%
\special{sh 1}%
\special{ar 2000 2520 16 16 0 6.2831853}%
\put(4.0000,-16.0000){\makebox(0,0){\footnotesize $x_{k-1}$}}%
\put(10.0000,-25.5000){\makebox(0,0){\footnotesize $x_k$}}%
\put(20.0000,-26.7000){\makebox(0,0){\footnotesize $x_{k+1}$}}%
\end{picture}}%
}
		      	\label{均衡点 of $n$vendor 1}
		      \end{center}\vspace{5mm}
		      
		\item $x_{k-1}=x_k$\quad(resp. $x_k=x_{k+1}$)\\
		\begin{center}
			\scalebox{1.0}{
{\unitlength 0.1in%
\begin{picture}(27.9000,22.7000)(-1.9000,-26.0500)%
%
\put(16.0000,-16.0000){\makebox(0,0)[lb]{}}%
%
\special{pn 8}%
\special{ar 1600 1600 1000 1000 0.0000000 6.2831853}%
%
\special{pn 13}%
\special{pa 1600 500}%
\special{pa 1600 700}%
\special{fp}%
\put(16.0000,-4.0000){\makebox(0,0){\footnotesize $0(=1)$}}%
%
\special{pn 4}%
\special{sh 1}%
\special{ar 600 1600 16 16 0 6.2831853}%
\special{sh 1}%
\special{ar 1000 2400 16 16 0 6.2831853}%
\special{sh 1}%
\special{ar 2000 2520 16 16 0 6.2831853}%
\put(4.0000,-16.0000){\makebox(0,0){\footnotesize $x_{k-2}$}}%
\put(9.7000,-25.8000){\makebox(0,0){\footnotesize $x_{k-1}=x_k$}}%
\put(20.0000,-26.7000){\makebox(0,0){\footnotesize $x_{k+1}$}}%
\end{picture}}%
}
			\label{均衡点 of $n$vendor 2}
		\end{center}\vspace{5mm}
		      If $|I_{k-1}|=0$, since $|I_{k-2}|+|I_{k-1}|\geq|\bar{I}|$ and $|I_{k-1}|+|I_k|\geq|\bar{I}|$,   we have $|I_{k-2}|=|I_k|=|\bar{I}|$.
		      (resp. If $|I_k|=0$,  since $|I_{k-1}|+|I_k|\geq|\bar{I}|$ and $|I_k|+|I_{k+1}|\geq|\bar{I}|$,  we have $|I_{k-1}|=|I_{k+1}|=|\bar{I}|$.) 
		      Therefore, we obtain $f_k({\bm x})=\dfrac{1}{2\cdot2}(|\bar{I}|+|\bar{I}|)=\dfrac{1}{2}|\bar{I}|$ in this case.
    	\end{enumerate}
		From (i) and (ii), we conclude that $f_k({\bm x})\geq \dfrac{1}{2}|\bar{I}|$.

\newpage
    Next, we show that $f_k({\bm x})\geq f_k(x_k\rightarrow[0,1])$.
	\begin{enumerate}
		\renewcommand{\labelenumi}{(\roman{enumi})}
		\item $x_k\rightarrow(x_{k-1},x_{k+1})$\quad($1\leq k\leq n$)
		\begin{enumerate}
			\renewcommand{\labelenumii}{(\alph{enumii})}
			\item $x_{k-1}\neq x_k$ and $x_k\neq x_{k+1}$\\
			We have $f_k(x_k\rightarrow(x_{k-1},x_{k+1}))=f_k({\bm x})$.
			\item $x_{k-1}=x_k$ or $x_k=x_{k+1}$\\
			We have $f_k(x_k\rightarrow(x_{k-1},x_{k+1}))\leq \dfrac12|\bar{I}|\leq f_k({\bm x})$.
			
		\end{enumerate}
		\item $x_k\rightarrow(x_l,x_{l+1})$\quad($x_l<x_{l+1}$ and $l\neq k-1,k$ and $1\leq l\leq n$)\\
		We have $f_k(x_k\rightarrow(x_l,x_{l+1}))=\dfrac12|I_l|\leq \dfrac12|\bar{I}|\leq f_k({\bm x})$.
		\item $x_k\rightarrow\{x_l\}$\quad($1\leq k\leq n$ and $1\leq l \leq n$)
		\begin{enumerate}
			\renewcommand{\labelenumii}{(\alph{enumii})}
			\item $l=k$\\
			It clearly holds that $f_k(x_k\rightarrow\{x_k\})=f_k({\bm x})$.
			
			\item The case where only vendor $k$ occupies at $x_k$ and moves next to $x_k$\vspace*{0.5ex}\\
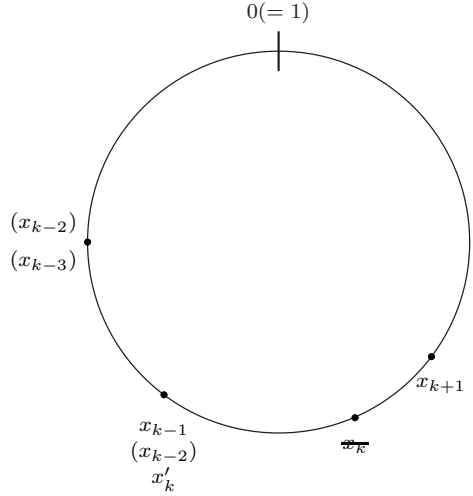
		\begin{center}
			\scalebox{1.0}{
{\unitlength 0.1in%
\begin{picture}(27.2500,24.0000)(-1.2500,-27.3500)%
%
\put(16.0000,-16.0000){\makebox(0,0)[lb]{}}%
%
\special{pn 8}%
\special{ar 1600 1600 1000 1000 0.0000000 6.2831853}%
%
\special{pn 13}%
\special{pa 1600 500}%
\special{pa 1600 700}%
\special{fp}%
\put(16.0000,-4.0000){\makebox(0,0){\footnotesize $0(=1)$}}%
%
\special{pn 4}%
\special{sh 1}%
\special{ar 600 1600 16 16 0 6.2831853}%
\special{sh 1}%
\special{ar 1000 2400 16 16 0 6.2831853}%
\special{sh 1}%
\special{ar 2000 2520 16 16 0 6.2831853}%
\put(3.7000,-15.0000){\makebox(0,0){\footnotesize ($x_{k-2}$)}}%
\put(3.7000,-17.0000){\makebox(0,0){\footnotesize ($x_{k-3}$)}}%
\put(10.0000,-25.8000){\makebox(0,0){\footnotesize $x_{k-1}$}}%
\put(20.0000,-26.7000){\makebox(0,0){\footnotesize \sout{$x_{k}$}}}%
%
\special{pn 4}%
\special{sh 1}%
\special{ar 2400 2200 16 16 0 6.2831853}%
\put(24.5000,-23.5000){\makebox(0,0){\footnotesize $x_{k+1}$}}%
\put(10.0000,-27.0000){\makebox(0,0){\footnotesize ($x_{k-2}$)}}%
\put(10.0000,-28.4000){\makebox(0,0){\footnotesize $x_k'$}}%
\end{picture}}%
}
			\label{均衡点 of $n$vendor 3}
		\end{center}\vspace{5mm}
			In this case, we deduce that $f_k(x_k\rightarrow\{x_l\})\leq\dfrac{1}{2\cdot2}(|\bar{I}|+2f_k({\bm x}))\leq\dfrac{1}{4}(2f_k({\bm x})+2f_k({\bm x}))=f_k({\bm x})$.

			\item The case where more than two vendors occupy at $x_k$ and vendor $k$ moves next to $x_k$, or 
			vendor $k$ moves to a location that is not next to $x_k$\vspace*{0.5ex}\\
		   In this case, we deduce that $f_k(x_k\rightarrow\{x_l\})\leq\dfrac{1}{2\cdot2}(|\bar{I}|+|\bar{I}|)=\dfrac12|\bar{I}|\leq f_k({\bm x})$.
		\end{enumerate}
	\end{enumerate}
	From (i)--(iii), we conclude that $f_k({\bm x})\geq f_k(x_k\rightarrow[0,1])$. Therefore we have showed that (1.1) is a sufficient condition for equilibrium.

	Next, we show that \eqref{eqn:main1} is a necessary condition for equilibrium. Therefore, we prove that if \eqref{eqn:main1} does not hold, then ${\bm x}$ is not equilibrium. 
	That is, assuming   there exists $k$ such that  $|I_k|+|I_{k+1}|<|\bar{I}|$, we prove that ${\bm x}$ is not equilibrium.
	Observing Proposition \ref{pro:1}, it suffices to consider the case where only one  vendor occupies a location.
	
	\begin{enumerate}
		\renewcommand{\labelenumi}{(\roman{enumi})}
		\item $x_k\neq x_{k+1}$ and $x_{k+1}\neq x_{k+2}$\\
		We have $f_{k+1}({\bm x})=\dfrac12(|I_k|+|I_{k+1}|)<\dfrac12|\bar{I}|$.
		Let $(x_j,x_{j+1})$ be a maximum open interval such that $|\bar{I}|=|(x_j,x_{j+1})|$.
		If $x_{k+1}\rightarrow(x_j,x_{j+1})$, we have $f_{k+1}(x_{k+1}\rightarrow(x_j,x_{j+1}))=\dfrac12|\bar{I}|>f_{k+1}({\bm x})$.\\
		Thus ${\bm x}$ with this case is not equilibrium.

		\item $x_k=x_{k+1}$\quad(resp. $x_{k-1}=x_k$)\\
		Since $|I_k|+|I_{k+1}|<|\bar{I}|$ and $|I_k|=0$, we drive $|I_{k+1}|<|\bar{I}|$. Forthermore, we notice $|I_{k-1}|\leq|\bar{I}|$.
		Thus we deduce that $f_k({\bm x})=\dfrac{1}{2\cdot2}(|I_{k-1}|+|I_{k+1}|)<\dfrac{1}{2\cdot2}(|\bar{I}|+|\bar{I}|)=\dfrac12|\bar{I}|$. Let $(x_j,x_{j+1})$ be a maximum open interval such that $|\bar{I}|=|(x_j,x_{j+1})|$.
		If $x_k\rightarrow(x_j,x_{j+1})$, we have $f_k(x_k\rightarrow(x_j,x_{j+1}))=\dfrac12|\bar{I}|>f_k({\bm x})$.
		Thus ${\bm x}$  is not equilibrium in this case.
		
	\end{enumerate}
	From (i)--(ii), we have showed that \eqref{eqn:main1} is a necessary condition for equilibrium.
	
	Therefore, we can complete the proof of Theorem \ref{thm:main}.
\end{proof}


\end{document}